\documentclass[conference]{IEEEtran}
\IEEEoverridecommandlockouts
\usepackage[letterpaper, 
top=0.75in, 
bottom=1.1in, 
left=0.64in, 
right=0.64in, 
columnsep=0.3in]{geometry}

\usepackage{lipsum} % For dummy text

\usepackage{cite}
\usepackage{amsmath,amssymb,amsfonts}
\usepackage{algorithmic}
\usepackage{graphicx}
\usepackage{textcomp}
\usepackage{xcolor}
\usepackage[colorlinks]{hyperref}
\usepackage{amsmath}
\usepackage{amsfonts}
\usepackage{amssymb}
\usepackage{amsthm}
\usepackage{bm}
\usepackage{tabularx}
\usepackage{mathrsfs} 
\usepackage{algorithmic}
\usepackage{algorithm} 
\usepackage{mathtools}
\usepackage{url}
\usepackage{hyperref}
\newtheorem{theorem}{Theorem}

\usepackage{stfloats}
\usepackage{float}
\usepackage{graphicx}
\usepackage{xcolor}
\usepackage{subfigure}

\makeatletter
\def\blfootnote{\xdef\@thefnmark{}\@footnotetext}
\makeatother
\def\BibTeX{{\rm B\kern-.05em{\sc i\kern-.025em b}\kern-.08em
		T\kern-.1667em\lower.7ex\hbox{E}\kern-.125emX}}

\begin{document}
	\title{\LARGE RIS-Aided Backscattering Tag-to-Tag Networks: Performance Analysis}
	
	\author{
		\IEEEauthorblockN{
			Masoud Kaveh\IEEEauthorrefmark{2},
			Farshad Rostami Ghadi\IEEEauthorrefmark{4},
			Zheng Yan\IEEEauthorrefmark{3}, and
			Riku J\"antti\IEEEauthorrefmark{2}}
		\IEEEauthorblockA{\IEEEauthorrefmark{2}Department of Information and Communications Engineering, Aalto University, Espoo, Finland}
		\IEEEauthorblockA{\IEEEauthorrefmark{4}Department of Electronic and Electrical Engineering, University College London, London, UK}
		\IEEEauthorblockA{\IEEEauthorrefmark{3}School of Cyber Engineering, Xidian University, Xi'an, China}
		
		%\thanks{This work is supported in part by the Academy of Finland under Grants 345072 and 350464.}
	}
	\maketitle
\begin{abstract}
Backscattering tag-to-tag networks (BTTNs) represent a passive radio frequency identification (RFID) system that enables direct communication between tags within an external radio frequency (RF) field. 
However, low spectral efficiency and short-range communication capabilities, along with the ultra-low power nature of the tags, create significant challenges for reliable and practical applications of BTTNs.
To address these challenges, this paper introduces integrating an indoor reconfigurable intelligent surface (RIS) into BTTN and studying RIS's impact on the system's performance. 
To that end, we first derive compact analytical expressions of the probability density function (PDF) and cumulative distribution function (CDF) for the received signal-to-noise ratio (SNR) at the receiver tag by exploiting the moment matching technique. Then, based on the derived PDF and CDF, we further derive analytical expressions of outage probability (OP), bit error rate (BER), and average capacity (AC) rate. Eventually, the Monte Carlo simulation is used to validate the accuracy of the analytical results, revealing that utilizing RIS can greatly improve the performance of BTTNs in terms of AC, BER, OP, and coverage region relative to traditional BTTNs setups that do not incorporate RIS.
\end{abstract}
\vspace{6pt}

\section{Introduction}
\IEEEPARstart{I}{n} today’s increasingly connected world, the demand for energy-efficient communication technologies and zero energy devices is growing, particularly in applications where regular battery replacement or recharging is impractical \cite{zero6G1}. 
%This need has catalyzed the development of passive communication systems that operate with minimal energy requirements \cite{zero6G1}. 
Radio frequency identification (RFID) technology stands out as a prime example of this innovation. RFID systems harness electromagnetic fields to identify and track tags attached to objects, without the need for an internal power source. The tags, powered by energy harvested from the reader’s signal, respond by backscattering a signal which is modulated with the necessary information \cite{BCFAS1}. 
%This capability not only eliminates the need for batteries but also enhances operational efficiency and accuracy across numerous industries, from supply chain management to access control \cite{BCFAS1}.
%
A significant innovation in this arena is backscatter communication (BC), a method that enables devices to communicate by reflecting existing RF signals instead of producing new ones. This approach greatly reduces power usage, making it perfectly suited for passive RFID systems \cite{BCRFID1,BCArbit1}. 
%In BC, a tag modifies the incoming RF signal by changing its reflective properties, embedding information that can then be captured and interpreted by a reader. 
While RFID adn BC technologies offer numerous benefits, these systems exhibit several drawbacks, especially in facilitating direct communication between tags. 

Typically, conventional BC and RFID setups only support interactions between a reader and a tag, which constrains the systems’ scalability and adaptability in various applications \cite{BCSI1}, 
%For example, in large-scale inventory management, smart environments, and the Internet of Things (IoT), 
where there is a growing need for tags to communicate directly with each other 
%to share information and coordinate actions 
without the constant intervention of a central reader
%. This need has led to the exploration of tag-to-tag communication within RFID systems 
\cite{BTTN-design-1}.
Backscattering tag-to-tag networks (BTTNs) are an innovative subset of passive RFID systems that facilitate direct tag-to-tag communication within an external RF field \cite{design-BTTN-2}. The basic setup of a BTTN consists of a \textit{talker tag} (TT), a \textit{listener tag} (LT), and an external RF source. 
%Utilizing the existing RF environment, the system enables communication by modulating and backscattering RF signals between the tags.
In this setup, TT modulates the RF field by altering its load impedance, which changes the characteristics of the backscattered signal directed towards the LT. The LT, in turn, demodulates the received backscattered signal using its envelope detector. Both tags are equipped with simple energy harvesting circuits, which allow them to charge their batteries by converting energy from the RF field, thus maintaining their operational capabilities without the need for external power sources \cite{design-BTTN-3}.

%Typical BTTNs face several significant challenges that must be overcome to ensure efficient and reliable operation. One of the primary challenges is the requirement for ultra-low power consumption. Since BTTNs rely on passive tags that harvest energy from the incident RF signals, it is crucial to minimize the power consumption to extend the operational lifetime of the tags and maintain their functionality.
%Another challenge is the low spectral efficiency of BC. The process of modulating and reflecting the RF signal inherently limits the bandwidth available for communication, resulting in lower data rates compared to active transmission methods. This issue is exacerbated in BTTNs, as the receiver is not a powerful reader but another tag with comparatively fewer processing resources.
%Furthermore, BTTNs are constrained by very short communication ranges. The passive nature of the tags and the reliance on backscattered signals result in a limited communication distance, which can be a significant drawback in applications requiring long-range interactions. 

Current BTTNs face multiple challenges that affect their efficiency and reliability. A primary constraint is the inherently low spectral efficiency of BC \cite{BTTN-performance-1}. The process of modulating and reflecting RF signals restricts the bandwidth available for data transmission, consequently limiting the data rates achievable within BTTNs. This issue is further exacerbated in environments where both the transmitter and receiver are passive tags, which possess significantly less processing power compared to conventional readers.
Furthermore, the operational range of BTTNs is notably limited due to the passive nature of the system. Tags rely entirely on backscattered signals for communication, which inherently weakens with distance, thus curtailing the effective communication range. This limitation poses substantial challenges in applications requiring longer-range interactions \cite{passive-BTTN-1}.
Additionally, the ultra-low power requirements of BTTNs compound these challenges. As the networks depend on passive tags that harvest energy from incidental RF signals, minimizing power consumption is crucial. It not only prolongs the operational lifespan of the tags but also ensures their continuous functionality. However, the passive operation and resource constraints of the tags make it difficult to compensate for the aforementioned spectral and range limitations, thereby demanding innovative solutions to enhance system performance.
To address these challenges, this paper proposes the integration of an indoor reconfigurable intelligent surface (RIS) into BTTNs and examines its impact on system performance. The contributions of this paper are as follows:
\begin{itemize}
    \item This paper is the first to explore the impact of RIS on the performance of BTTNs, analyzing how various system parameters in RIS-aided BTTNs affect system performance.
    \item To that end, we derive compact closed-form expressions for the probability density function (PDF) and cumulative distribution function (CDF) of the received signal-to-noise ratio (SNR) at LT, using the moment matching technique. This approach enables us to develop a simple yet insightful analytical framework that aptly describes the behavior of this complex system model.
    Then, we establish comprehensive analytical frameworks for average capacity (AC) rate, outage probability (OP), and bit error rate (BER) based on the derived PDF and CDF.
    \item The accuracy of our analytical models is validated through Monte Carlo simulations. The results affirm that the use of RIS significantly enhances various performance metrics in BTTNs compared to conventional settings without RIS.
\end{itemize}

The remainder of this paper is structured as follows. Section \ref{sec_sys} details the system and channel models for RIS-aided BTTNs. Section \ref{sec_performance} discusses the statistical analysis and elaborates on the derivation of closed-form expressions utilized in this study. Section \ref{sec_sim} presents our experimental setup and the corresponding results. Finally, Section \ref{sec_conclusion} summarizes the conclusions drawn from this paper.

\section{System and Channel Models} \label{sec_sys}
\begin{figure}[t]
    \centering    \includegraphics[width=0.43\textwidth]{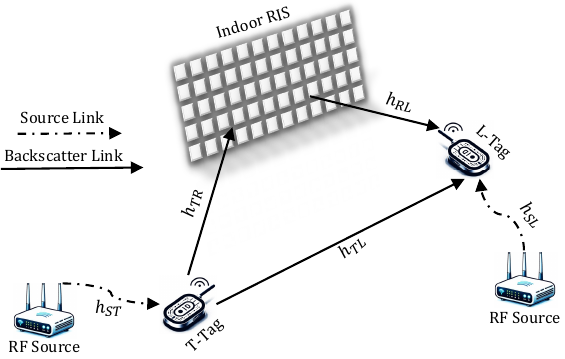}
    \caption{System model: A RIS-aided BTTN.}
    \label{fig:sysmodel}
\end{figure}
Fig. \ref{fig:sysmodel} illustrates the system model for a RIS-aided BTTN. In this setup, TT and LT function as semi-passive backscatter devices, powered by continuous wave carrier signals emitted by RF sources in an indoor setting. TT’s objective is to transmit its messages to LT, facilitated by an indoor RIS equipped with \emph{N} reflective elements. In standard BTTN operations, TT alters its impedance, affecting the tag antenna’s reflection coefficient and thus encoding the information within the backscattered signal, which is subsequently captured by LT \cite{BTTN-CSI-1}.
The signal LT receives is a mix of the direct, unmodulated signal from the RF source and the modulated backscattered signal from TT, via both direct and RIS-assisted pathways.
For simplicity, and assuming no loss of generality, it is noted that the RF source transmits unmodulated carriers, allowing both tags to apply cancellation techniques to lessen interference stemming from the source’s link \cite{Riku_DirPathInterf, Riku_IEEERFID1}.
It is presumed that the RIS is aware of the cascade link’s channel state information (CSI), enabling it to optimize the phase shifting coefficients of its elements to maximize the SNR at LT
%\footnote{This process can be executed at the reader, which involves transmitting pilot symbols from BDs through different phase configurations of the RIS. This allows for a linear estimation of the cascaded channel by aggregating the received signals that correspond to each configuration. Then, the estimated channel coefficients are communicated back to the RIS controller \cite{BjornsonCSI}.}
\cite{BjornsonCSI, RIS-PLS-SG}. 
We also assume that TT and LT are equipped with a single antenna for simplicity. Therefore, the received signal at TT can be given as follows.
\begin{align} % Eq(1)
y_\text{T}=\sqrt{P_\text{s}}h_\text{ST}+n_\text{T},
\end{align} where $P_\text{s}$ is the power of the signal emitted by the RF source, $h_\text{ST}$ is the channel coefficient between source and TT, and $n_\text{T}$ denotes the additive white Gaussian noise (AWGN) at TT with zero mean and variances $\sigma^2_\text{T}$. Since the noise power caused by a tag's antenna is considerably smaller than the received signal from the source, we neglect it in the rest of this paper \cite{ref57}. 
Therefore, the backscattered signal received at LT can be expressed as:
\begin{align} \label{eq-LT-sig}
y_\text{L} = \lambda_\text{T} S(t) \sqrt{P_\text{s}} \left( h_\text{ST} h_\text{TL} + h_\text{ST} \mathbf{H}_\mathrm{RL}^{\mathcal{T}} \mathbf{\Phi}  \mathbf{H}_\mathrm{TR} \right) +n_\text{L} ,
\end{align}
in which  
\( \lambda_\text{T} \) denotes the backscattering coefficient of TT,
    \( h_\text{TL} \) is the channel coefficient between TT and LT,  $n_\text{L}$ shows the noise at LT with zero mean and variances $\sigma^2_\text{T}$,
$S(t)$ defines the information signal backscattered from TT with a unit power, $\mathbf{H}^{\mathcal{T}}$ represents the transpose of matrix $\mathbf{H}$, $\mathbf{H}_\mathrm{TR}$ is the TT-to-RIS channel coefficient, and $\mathbf{H}_\mathrm{RL}$ is the RIS-to-LT channel coefficient. 
The term $\mathbf{\Phi}$ represents the adjustable phase shift matrix of RIS for maximizing the received signal power at LT, which can be defined as $\mathbf{\Phi}=\text{diag}\left(\left[\mathrm{e}^{j\phi_1}, \mathrm{e}^{j\phi_2},...,\mathrm{e}^{j\phi_N}\right]\right)$. Therefore, the RIS-aided channel coefficients like $\mathbf{H}_\mathrm{TR}$ and $\mathbf{H}_\mathrm{RL}$ contain the \emph{N} channel coefficients from TT to RIS and from RIS to LT as $\mathbf{H}_\mathrm{TR}=d_\mathrm{TR}^{-\chi}.\left[h_\mathrm{TR_1}\mathrm{e}^{-j\delta_1}, h_\mathrm{TR_2}\mathrm{e}^{-j\delta_2},..., h_\mathrm{TR_N}\mathrm{e}^{-j\delta_N}\right]$ and $\mathbf{H}_{RL }=d_\mathrm{RL}^{-\chi}.\left[h_\mathrm{RL_1}\mathrm{e}^{-j\zeta_{1}}, h_\mathrm{RL_2}\mathrm{e}^{-j\zeta_{2}},..., h_\mathrm{RL_N}\mathrm{e}^{-j\zeta_{N}}\right]$, where $j=\sqrt{-1}$, and $d_\mathrm{{TR}}$ and $d_\mathrm{{RL}}$ denote the distance between TT and RIS and the distance between RIS and LT, respectively \cite{V2V-RIS-TIV}. 
The term $\chi$ indicates the path-loss exponent, terms $h_\mathrm{TR_n}$ and $h_\mathrm{RL_n}$ are the amplitudes of the corresponding channel coefficients, and the terms $\mathrm{e}^{-j\delta_n}$ and $\mathrm{e}^{-j\zeta_{n}}$ denote the phase of the respective links for $n\in\left\{1,2,...,  N\right\}$. 
As there may be obstructions impeding the direct link in BTTNs, we also assume that all links follow Rayleigh fading distribution. This received signal is then processed by LT's demodulator circuit to extract the information sent by TT.
If LT needs to communicate with TT, the process is similar but in reverse. LT modulates its impedance to backscatter the RF signal, and TT receives this modulated backscattered signal.

\section{Performance Analysis} \label{sec_performance}
In this section, we initially perform a statistical analysis of the system model under study. Subsequently, we develop closed-form expressions for the PDF and CDF of the received SNR at LT. Following this, we formulate concise analytical frameworks for OP, BER, and AC.

\subsection{Statistical Analysis} \label{sec:SNR}
This section presents an analysis of the received SNR at LT, followed by the derivation of concise analytical expressions for the PDF and CDF, based on the evaluated SNR and moment matching technique.
According to \eqref{eq-LT-sig}, the instantaneous SNR at LT can be determined as
\newcommand\myeq{\mathrel{\stackrel{\makebox[0pt]{\mbox{\normalfont\small def}}}{=}}}
\newcommand\myapprox{\mathrel{\stackrel{\makebox[0pt]{\mbox{\normalfont\footnotesize (a)}}}{\approx}}}
\begin{align} 
&\gamma_L=\frac{{\left|\sqrt{P_\text{s}} \left( h_\text{ST} h_\text{TL} + h_\text{ST} \mathbf{H}_\mathrm{RL}^{\mathcal{T}} \mathbf{\Phi}  \mathbf{H}_\mathrm{TR} \right)\right|}^2}{n_\text{L}}  \\ 
&\hspace{-0.2cm}\approx\frac{P_\text{s}|h_\mathrm{ST}|^2|h_\mathrm{TL}|^2}{{d}^\chi_\mathrm{ST}{d}^\chi_\mathrm{TL}{\sigma}^2_\mathrm{L}}
\hspace{-2pt} +\hspace{-2pt}
\frac{P_\mathrm{s}|h_\mathrm{ST}|^2\left|\sum_{n=1}^{N}h_\mathrm{TR_n} h_\mathrm{RL_n} \mathrm{e}^{j\Psi_n}\right|^2}{{d}^\chi_\mathrm{ST}{d}^\chi_\mathrm{TR}{d}^\chi_{RL}{\sigma}^2_\mathrm{L}}  \\ 
&\overset{(a)}{=} \bar{\gamma}_\text{x} \underbrace{\left|h_\mathrm{ST}\right|^2}_{X_1}  \underbrace{\left|h_\mathrm{TL}\right|^2}_{X_2} + \bar{\gamma}_\text{y} \underbrace{\left|h_\mathrm{ST}\right|^2}_{Y_1} \underbrace{\left|\sum_{n=1}^{N} h_\mathrm{TR_n} h_\mathrm{RL_n} \right|^2} _{Y_2} ,\label{gamma_Reader_2}
\end{align}
where $(a)$ results from ideal phase shifting in the RIS \cite{ref41,ref3941}, $\Psi_n= \phi_n-\delta_n-\zeta_{n}$, $\bar{\gamma}_\text{x}= \frac{\gamma_0}{{d}^\chi_\mathrm{ST}{d}^\chi_\mathrm{TL}}$, $\bar{\gamma}_\text{y}=\frac{\gamma_0}{{d}^\chi_\mathrm{ST}{d}^\chi_\mathrm{TR}{d}^\chi_{RL}}$, and $\gamma_0 = \frac{P_\text{s}}{\sigma^2_\text{L}}$. 
$X_1$ and $X_2$ are exponentially distributed random variables (RVs) with $\mathbb{E}[X_1]=\alpha$, $\mathbb{V}(X_1)=\alpha^2$, $\mathbb{E}[X_2]=\beta$, and $\mathbb{V}(X_2)=\beta^2$, where $\mathbb{E}[\cdot]$ and $\mathbb{V}[\cdot]$ denote the mean and variance, respectively.
The phase information in tags can be measured by multi-phase probing (MPP)-based technique as discribed in \cite{BTTN-CSI-1}.
By defining $X = X_1 X_2$ as the product of independent direct and source links, the mean and variance of $X$ are calculated as $\mathbb{E}[X] = \alpha\beta$ and $\mathbb{V}(X) = 3\alpha^2\beta^2$, respectively, according to the following equations. 
\begin{align} \label{mean-pro-fromula}
    \mathbb{E}[X_1 X_2]= \mathbb{E}[X_1] \mathbb{E}[X_2],
\end{align}
\begin{align} \label{var-pro-fromula}
    \mathbb{V}(X_1 X_2) \hspace{-2pt}= \hspace{-2pt} \mathbb{V}(X_1)\mathbb{V}(X_2) \hspace{-2pt} + \hspace{-2pt} \mathbb{V}(X_1) (\mathbb{E}[X_2])^2 \hspace{-2pt} + \hspace{-2pt} \mathbb{V}(X_2) (\mathbb{E}[X_1])^2.
\end{align}
Therefore, the RV $\gamma_\text{x}$ representing the received SNR at LT from the direct link, $\gamma_\text{x}=\bar{\gamma}_\text{x} X$ has mean and variance as: 
\begin{align}
    \mathbb{E}[\gamma_\text{x}] &= \bar{\gamma}_\text{x} \alpha\beta, \\    \mathbb{V}(\gamma_\text{x}) &= 3\bar{\gamma}^2_\text{x} \alpha^2\beta^2.
 \end{align}
 
For the RIS-aided part of the received SNR, $Y_1$ has the same mean and variance values as $X_1$. On the other hand, $Y=\bar{\gamma}_y Y_2$ follows a Gamma distribution with PDF and CDF as \cite{Gamma_dist_varmean,Gamma_dist_sum}:
\begin{align}
   f_{Y}(y) &= \frac{y^{\frac{Nk'-2}{2}} \mathrm{e}^{-\frac{\sqrt{y}}{\theta'}}}{2\theta'^{Nk'} \Gamma(Nk')} ,  \\ F_{Y}(y) &= \frac{\Upsilon\left(Nk', \frac{\sqrt{y}}{\theta'}\right)}{\Gamma(Nk')},    
\end{align}
 where \(\Upsilon(\cdot, \cdot)\) is the lower incomplete gamma function with the shape parameter $k' = \frac{N\pi^2}{16-\pi^2}$ and the scale parameter $\theta' = \frac{\left(16-\pi^2\right) \sqrt{\bar{\gamma}_y \delta_1 \delta_2}}{4\pi}$, in which $\delta_1$ and $\delta_2$ denote the mean of the channel gains between TT-to-RIS and RIS-to-LT, respectively. Besides, \(\Gamma(\cdot)\) denotes the gamma function. Therefore, since the mean and variance of the gamma distribution are defined as \(k'\theta'\) and \(k'\theta'^2\), $\mathbb{E}[Y]$ and $\mathbb{V}[Y]$ can be respectively shown as follows: 
 \begin{align}
    \mathbb{E}[Y] &= \frac{N \pi \sqrt{\bar{\gamma}_y \delta_1 \delta_2}}{4}, \\
     \mathbb{V}[Y] &= \frac{N \bar{\gamma}_y \delta_1 \delta_2 \left(16-\pi^2\right)}{16}.
 \end{align}
Now, by defining $\gamma_\text{y}=Y_1 Y$, and using \eqref{mean-pro-fromula} and \eqref{var-pro-fromula}, $\gamma_\text{y}$ has the following mean and variance, respectively
\begin{align} \label{E-gamma-y}
    \mathbb{E}[\gamma_\text{y}] &= \frac{\alpha N \pi \sqrt{\bar{\gamma}_y \delta_1 \delta_2}}{4}, \\
     \mathbb{V}[\gamma_\text{y}] &= \alpha^2 N \bar{\gamma}_y \delta_1 \delta_2 \left( 2+ \frac{(N-2)\pi^2}{16} \right) , \label{V-gamma-y}
 \end{align}
%\begin{align} \label{E-gamma-y}
%\mathbb{E}[\gamma_\text{y}] &= \bar{\gamma}_y \alpha N \left(1 + \frac{(N-1)\pi^2}{16}\right), \\
%\mathbb{V}(\gamma_\text{y}) \hspace{-2pt} &= \hspace{-2pt} \bar{\gamma}_y^2 \alpha^2 \hspace{-3pt} \left( \hspace{-2pt} \frac{2A(N-A) \hspace{-2pt} + \hspace{-2pt} 3N\pi^2 A^2}{256 \pi^3} \hspace{-2pt} + \hspace{-2pt} N^2 \left(1 \hspace{-2pt} + \hspace{-2pt} \frac{(N-1)\pi^2}{16}\right)^2 \right). \label{V-gamma-y}
%\end{align}
%
Finally, given that the received SNR at LT is the summation of two independent RVs as defined in \eqref{gamma_Reader_2}, the mean and variance of $\gamma_\mathrm{L}$ can be obtained as:
\begin{align} \label{mean-gamma-L}
\mathbb{E}[\gamma_\mathrm{L}] &= \alpha \left(\bar{\gamma}_x \beta + \frac{ N \pi \sqrt{\bar{\gamma}_y \delta_1 \delta_2}}{4} \right), \\
\mathbb{V}(\gamma_\mathrm{L}) &= \alpha^2 \left(3{\bar{\gamma}}^2_x \beta^2 + N \bar{\gamma}_y \delta_1 \delta_2 \left(2+ \frac{ (N-2)\pi^2 }{16} \right) \right). \label{var-gamma-L}
\end{align}

\subsection{Outage Probability}
OP is a pivotal performance metric in communication systems, especially those operating over fading channels. It quantifies the likelihood that the channel's capacity falls below a certain threshold information rate, \(R_t\), necessary for satisfactory communication. Mathematically, if \(C\) represents the instantaneous channel capacity, OP can be defined as:
\begin{align} \label{op-def1}
P_{\text{out}} = \Pr(C \leq R_\text{t}),
\end{align}
where the threshold \(R_\text{t}\) is a positive value that typically corresponds to the minimum rate needed to achieve a certain quality of service or to fulfill specific application requirements. This definition implies that OP measures the probability that the channel cannot support a specified minimum rate, \(R_\text{t}\), required for adequate transmission quality. 
\begin{theorem}
The OP in the studied RIS-aided BTTN under Rayleigh fading channels follows a gamma distribution as:
\begin{align} \label{OP}
    P_{\mathrm{out}} &= \frac{\Upsilon\left(k, \frac{\gamma_\mathrm{th}}{\theta}\right)}{\Gamma(k)},
\end{align}
where $\gamma_\mathrm{th}=2^{R_\text{t}}-1$ denotes the SNR threshold. In addition, the shape parameter $k$ and the scale parameter $\theta$ are respectively given by %\eqref{cdf-theta} and \eqref{cdf-k}, respectively.
\begin{align} \label{cdf-k}
 k &= \frac{\left(4\bar{\gamma}_x \beta + N \pi \sqrt{\bar{\gamma}_y \delta_1 \delta_2}\right)^2} {16 \left(3 \bar{\gamma}_x^2 \beta^2 + N \bar{\gamma}_y \delta_1 \delta_2 \left(2+ \frac{ (N-2)\pi^2 }{16} \right) \right)},\\
 \label{cdf-theta}
\theta &= \frac{4\alpha \left(3 \bar{\gamma}_x^2 \beta^2 + N \bar{\gamma}_y \delta_1 \delta_2 \left(2+ \frac{ (N-2)\pi^2 }{16} \right) \right) }
{4\bar{\gamma}_x \beta + N \pi \sqrt{\bar{\gamma}_y \delta_1 \delta_2}}.
\end{align}
\end{theorem}
\begin{proof}
    First, we utilize the moment matching technique to approximate the corresponding PDF and CDF of the received SNR at LT \cite{Mom-Match,FAS-BC-ISAC}. The gamma distribution, characterized by just the two adjustable shape and scale parameters is frequently used to approximate complex distributions due to its simplicity. Consequently, the mean and variance of the approximate Gamma distribution are defined as \(k\theta\) and \(k\theta^2\), respectively. Thus, having established the mean and variance of \(\gamma_{\text{L}}\) as shown in \eqref{mean-gamma-L} and \eqref{var-gamma-L}, we can derive \(k\) and \(\theta\) as demonstrated in \eqref{cdf-theta} and \eqref{cdf-k}. This allows us to approximate \(f_{\gamma_{\text{L}}}(\gamma_{\text{L}})\) and \(F_{\gamma_{\text{L}}}(\gamma_{\text{L}})\) as detailed in \eqref{snr-pdf-gamma} and \eqref{snr-cdf-gamma}, respectively. 
\begin{align} \label{snr-pdf-gamma}
f_{\gamma_\mathrm{L}}(\gamma_\mathrm{L}) &= \frac{{\gamma_\mathrm{L}}^{k-1} \mathrm{e}^{-\frac{\gamma_\mathrm{L}}{\theta}}}{\Gamma(k) \theta^{k}} , \\
\label{snr-cdf-gamma}
    F_{\gamma_\mathrm{L}}(\gamma_\mathrm{L}) &= \frac{\Upsilon\left(k, \frac{\gamma_\mathrm{L}}{\theta}\right)}{\Gamma(k)},
\end{align}
%%%%%%%
On the other hand, the definition of OP in \eqref{op-def1} leads to the following  equation:
\begin{align}
P_{\text{out}} = F_{\gamma_\text{L}}(\gamma_{\text{th}}).
\end{align}
Now, by having the CDF of $\gamma_\text{L}$, the OP can be obtained as \eqref{OP}, thereby, completing the proof. 
\end{proof}

\subsection{Bit Error Rate}
In evaluating the reliability of the RIS-aided BTTNs, BER serves as a crucial metric, which quantifies the rate at which errors occur in a communication channel and is fundamentally influenced by the SNR. To rigorously assess the BER in our proposed system model, we derive an expression that integrates over $f_{\gamma_\mathrm{L}}(\gamma_\mathrm{L})$. This approach allows us to compute the average BER across varying channel conditions, providing a comprehensive measure of system performance. The average BER, defined for a gamma-distributed SNR with predefined parameters \(k\) and \(\theta\), is given by:
\begin{align} \label{ber-int-1}
BER &= \int_0^\infty Q \left(\sqrt{2\gamma_\mathrm{L}}\right) f_{\gamma_\mathrm{L}}(\gamma_\mathrm{L}) \, \text{d}\gamma_\mathrm{L}
\end{align}
where \(Q(\cdot)\) is the Q-function, representing the tail probability of the standard normal distribution. This formulation provides a detailed insight into the impact of SNR on the error performance, highlighting how RIS can mitigate errors by enhancing the effective SNR at LT.
\begin{theorem} \label{ThmBER}
    The BER in the studied RIS-aided BTTN under amplitude shift keying (ASK) modulation can be expressed as:
    \begin{align} \label{ber}
&BER \hspace{0pt}= \hspace{0pt} \frac{1}{ 2\sqrt{\pi} \ \Gamma(k) \ln{2}} \
 G^{2,1}_{2,2}\left( \hspace{-4pt} \begin{array}{c}
				\theta 
    \end{array} \hspace{-2pt}
			\Big\vert \hspace{-2pt} \begin{array}{c}
				1-k,1 \\
				0, \frac{1}{2} \\
			\end{array} \hspace{-4pt} \right), 
\end{align}
where $G^{m,n}_{p,q}(  . )$ denotes the Meiger's G function, and \( k \) and \( \theta \) are given by \eqref{cdf-theta} and \eqref{cdf-k}, respectively.
\end{theorem}
\begin{proof}
By having $Q(x) = \frac{1}{2} erfc \left(\frac{x}{\sqrt{2}} \right)$ using 
    the Meijer's G-function demonstration of $erfc(.)$ function as shown in \cite[Eq. 8.4.14.2]{ref59}, and inserting them together with \eqref{snr-pdf-gamma} into \eqref{ber-int-1}, we will have \eqref{ber-int-1} as:
    \begin{align} \label{ber-int-2}
BER \hspace{-2pt}=\hspace{-2pt} \frac{1}{ 2\sqrt{\pi} \Gamma(k) \theta^k \ln{2}}\int_{0}^{\infty} \hspace{-4pt} {\gamma_\mathrm{L}}^{k-1} \hspace{0pt} \mathrm{e}^{-\frac{\gamma_\mathrm{L}}{\theta} } 
 G_{1,2}^{2,0} \hspace{-2pt} \left( \hspace{-7pt} \begin{array}{c}
				\gamma_\mathrm{L} \end{array} \hspace{-4pt}
			\Big\vert \hspace{-4pt} \begin{array}{c}
				1\\
				0, \frac{1}{2}\\
			\end{array} \hspace{-6pt} \right) \hspace{-2pt}\mathrm{d}{\gamma_\mathrm{L}},        
    \end{align}
    %%%%%%%%
By using \cite[Eq. 2.24.3.1]{ref59}, we can obtain \eqref{ber-int-2} as \eqref{ber}, and thus, the proof is completed for Thm. \ref{ThmBER}.
\end{proof}

\subsection{Average Capacity}
AC in communication systems quantifies the mean achievable information rate over a channel subject to varying conditions, notably those influenced by fading characteristics. Unlike OP, which focuses on the likelihood of failing to meet a threshold information rate, AC reflects the expected performance of the channel across all potential state variations \cite{Ghadi123}. AC in a fading environment can be calculated from the PDF of the SNR by using an integral formula. This calculation typically assumes logarithmic utility (i.e., Shannon capacity) for the channel can be formulated as:
\begin{align} \label{Cbar-int1}
\bar{C} = \int_0^\infty \log_2(1 + \gamma_\mathrm{L}) \ f_{\gamma_\mathrm{L}}(\gamma_\mathrm{L}) \  \mathrm{d}\gamma_\mathrm{L}.
\end{align}
%%%%%
\begin{theorem} \label{ThmAC}
The AC for the studied RIS-aided BTTN under Rayleigh fading channels can be expressed as:
    \begin{align} \label{Cbar}
&\bar{C} \hspace{0pt}= \hspace{0pt} \frac{1}{ \Gamma(k) \ln{2}} \
 G^{1,3}_{3,2}\left( \hspace{-4pt} \begin{array}{c}
				\theta \end{array} \hspace{-2pt}
			\Big\vert \hspace{-2pt} \begin{array}{c}
				1-k,1,1\\
				1,0\\
			\end{array} \hspace{-4pt} \right), 
\end{align}
where $G^{m,n}_{p,q}(  . )$ denotes the Meiger's G function, and \( k \) and \( \theta \) are given by \eqref{cdf-theta} and \eqref{cdf-k}, respectively.
\end{theorem}
%%%%%%%
\begin{proof}
    Followed by the AC formula and having the Meijer's G-function demonstration of the logarithm function as shown in \cite[Eq. 8.4.6.5]{ref59}, we insert \eqref{snr-pdf-gamma} into \eqref{Cbar-int1}. Thus, we can re-write \eqref{Cbar-int1} as \eqref{Cbar_int}.
    as:
%\begin{align}
 %   \ln{(1+x)}= G^{1,2}_{2,2}\left( \hspace{-4pt} \begin{array}{c}
	%			x \end{array} \hspace{-2pt}
	%		\Big\vert \hspace{-2pt} \begin{array}{c}
	%			1,1\\
	%			1,0\\
	%		\end{array} \hspace{-4pt} \right).
%\end{align}
%%%%%%%%
\begin{align} \label{Cbar_int}
&\bar{C} \hspace{0pt}= \hspace{0pt} \frac{1}{ \Gamma(k) \theta^k \ln{2}}\int_{0}^{\infty} {\gamma_\mathrm{L}}^{k-1} \hspace{2pt} \mathrm{e}^{-\frac{\gamma_\mathrm{L}}{\theta} } 
 G_{2,2}^{1,2}\left( \hspace{-4pt} \begin{array}{c}
				\gamma_\mathrm{L} \end{array} \hspace{-2pt}
			\Big\vert \hspace{-2pt} \begin{array}{c}
				1,1\\
				1,0\\
			\end{array} \hspace{-4pt} \right) \mathrm{d}{\gamma_\mathrm{L}}, 
\end{align}
%%%%%%%%%%
By using \cite[Eq. 2.24.3.1]{ref59}, we can obtain \eqref{Cbar_int} as \eqref{Cbar}, and thus, the proof is completed for Thm. \ref{ThmAC}.
\end{proof}
%
%\subsection{Coverage Area} /In our system model, 
%
%%d_\mathrm{TL} =  \sqrt[\chi]{\frac{\gamma_0 \hspace{2pt} \beta } {\frac{ \left( 2^{\bar{C}} - 1\right) d^\chi_\mathrm{ST}} {\alpha} - \frac{\gamma_0 \hspace{1pt} N \left( 16 + (N-1) \pi^2 \right)}{16 \ d^x_\mathrm{TR} d^x_\mathrm{RL} }}}
%\end{align}
%\begin{align}
 %   d_\mathrm{TL} = \sqrt[\chi]{ \frac{\tilde{\gamma} \ln \frac{1}{1-P_{\text{out}}}}{\gamma_{\text{th}}} }
%\end{align}

\section{ Results and Discussions} \label{sec_sim}
In this section, we conduct extensive simulations to assess the performance RIS-aided BTTNs based on different system parameters. We also use Monte-Carlo simulation technique to verify the analytical expressions for OP, BER, and AC. 
\subsection{Simulation Setup}
In this section, we describe a RIS-aided BTTN system composed of passive talker and listener tags with extremely limited computational and energy capabilities, RF sources to power the tags (which may be ambient or dedicated), and a RIS equipped with $N$ reflective elements. We focus on an indoor setting where all components are stationary \cite{Indoor_RIS_02}, with distances set as $d_\mathrm{ST}= d_\mathrm{TR}= d_\mathrm{RL}= 1$m and $d_\mathrm{RL}=1.5$m. The system parameters include a rate threshold of $R_\text{t}=2$ bps/Hz, noise at the LT of $\sigma^2_L=-50$ dBm, a transmitted power from the RF source of $P_s=1$ dBm, and a path-loss factor of $\chi =3.5$ \cite{BCAuth}. Additionally, the MATLAB 2022a platform was utilized to develop our analytical frameworks and conduct Monte Carlo simulations for validation.

\subsection{Simulation Results}

\begin{figure}[t]
    \centering    \includegraphics[width=0.36\textwidth]{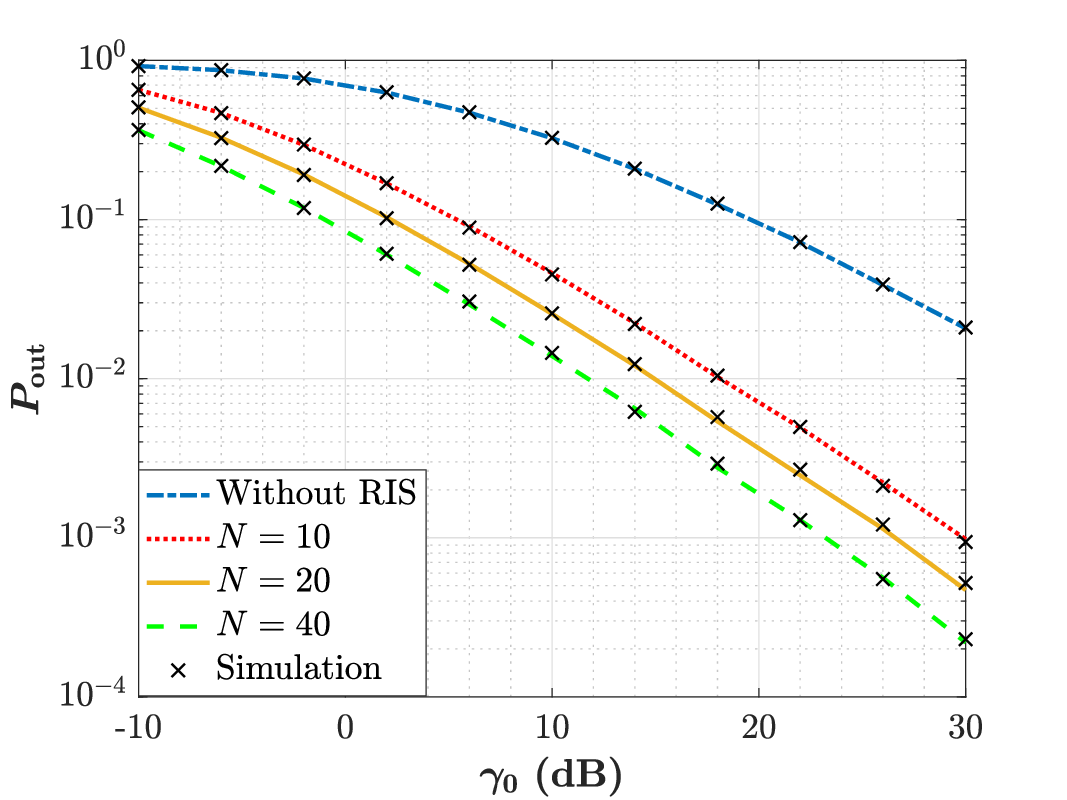}\vspace{-0.3cm}
    \caption{OP vs average SNR for different number of RIS reflecting elements.}
    \label{fig:OP_N}
\end{figure}
The impact of RIS on OP within BTTNs is rigorously analyzed under varied average SNR levels, as shown in Fig. \ref{fig:OP_N}. One can also observe that the utilization of RIS markedly improves communication reliability across all observed SNR regimes compared to the \textit{Without RIS} scenario. The performance enhancement is particularly notable with an increase in the number of RIS elements, highlighting the critical role of RIS in overcoming adverse channel conditions and boosting signal diversity within BTTNs. As more RIS elements are employed, there is a significant reduction in OP, effectively demonstrating how smart propagation environment reconfiguration can bolster wireless communication robustness in BTTNs. This positive trend is consistently supported by the tight correlation between the simulation results and theoretical predictions, achieved through the Monte Carlo simulation technique, thereby providing robust validation for the analytical models proposed for evaluating OP in BTTNs. 
\begin{figure}[t]
    \centering    \includegraphics[width=0.36\textwidth]{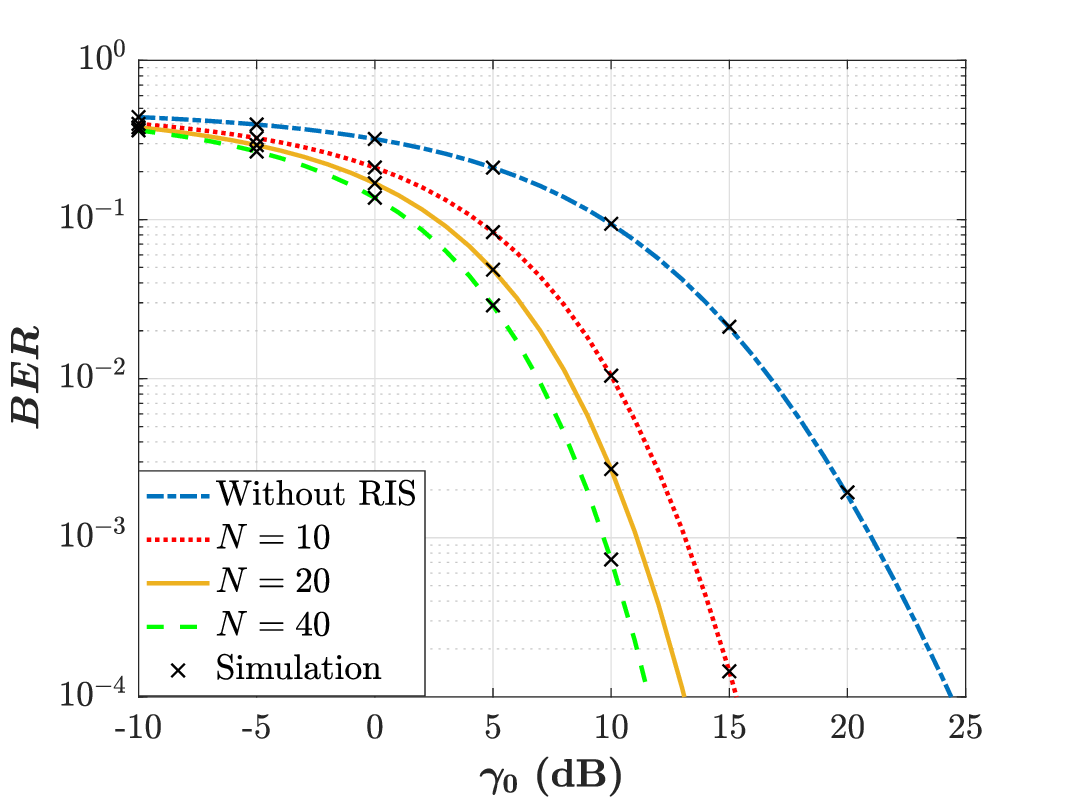}\vspace{-0.3cm}
    \caption{BER vs average SNR for different number of RIS reflecting elements.}
    \label{fig:BER_N}
\end{figure}
Fig. \ref{fig:BER_N} illustrates the impact of integrating RIS into BTTNs on the BER as a function of the average SNR for different RIS element configurations. The number of RIS elements, ranging from 10 to 40, shows a marked improvement in BER performance at all SNR levels. Notably, the presence of RIS significantly reduces the BER, particularly at higher SNR levels, compared to \textit{Without RIS} scenario. Additionally, the close alignment of the simulation points with the theoretical curves confirms the accuracy of the models used to predict BER in RIS-aided BTTNs.
\begin{figure}[t]
    \centering    \includegraphics[width=0.36\textwidth]{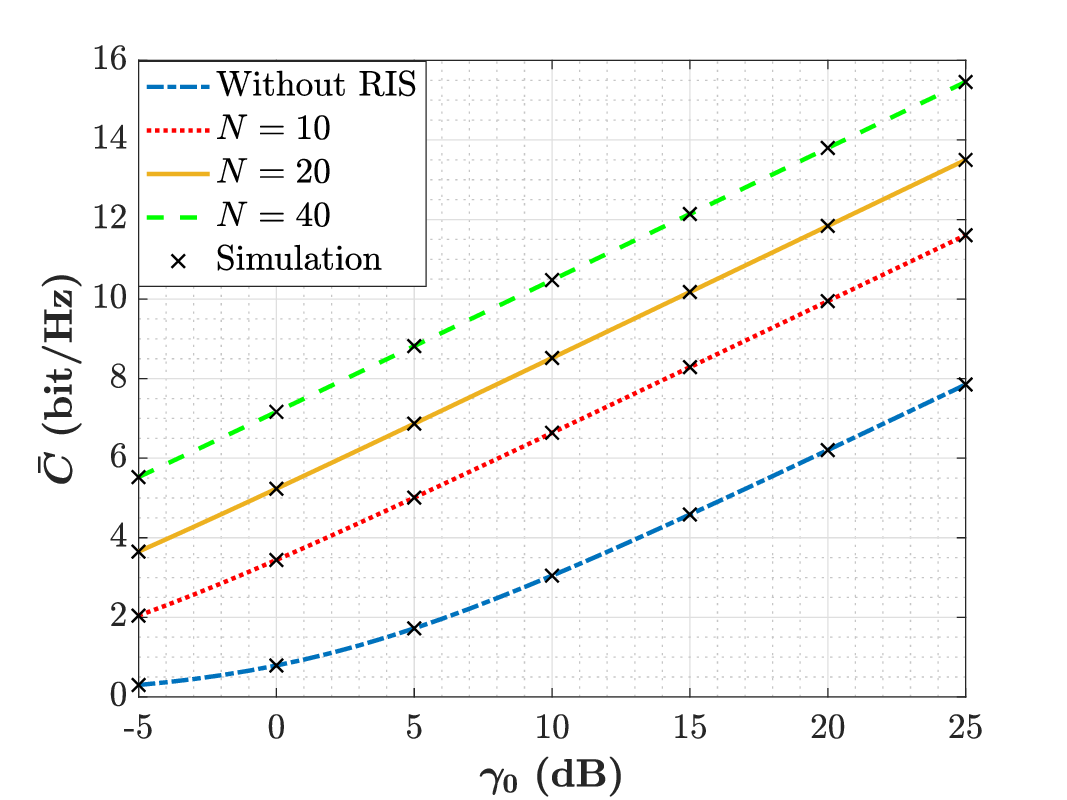}\vspace{-0.3cm}
    \caption{AC vs average SNR for different number of RIS reflecting elements.}
    \label{fig:AC_N}
\end{figure}
Fig. \ref{fig:AC_N} illustrates the enhancement of average channel capacity, denoted as $\bar{C}$ as a function of the average SNR, across various configurations within the studied RIS-aided BTTN. The data reveal a pronounced increase in channel capacity when RIS is integrated into the system. This trend is increasingly beneficial as the number of RIS elements escalates, with the configuration involving $N = 40$ achieving a capacity exceeding $9 \ \mathrm{bit/Hz}$ at an average SNR of approximately $5$ dB. This enhancement underscores the effectiveness of RIS in mitigating adverse channel conditions and boosting the overall spectral efficiency of BTTN communication systems. The simulation points exhibit close alignment with the theoretical curves, lending credence to the robustness and accuracy of the analytical models deployed to derive the closed-form expressions of AC.
\begin{figure}[t]
    \centering    \includegraphics[width=0.36\textwidth]{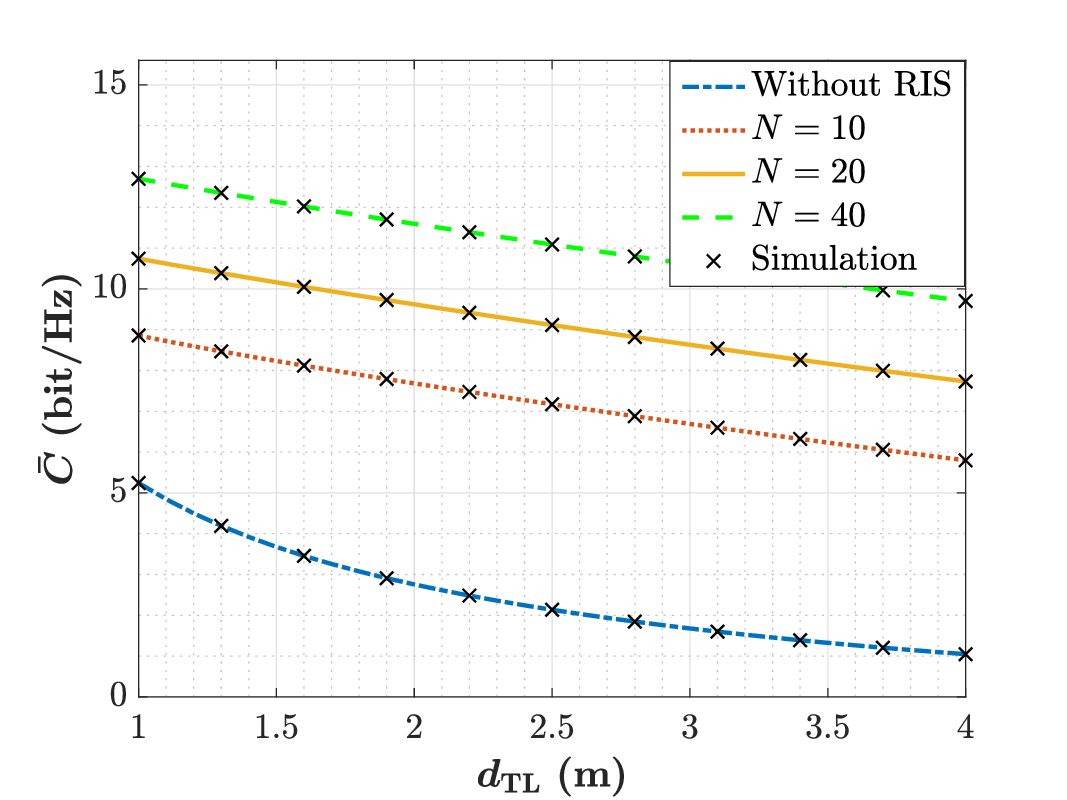}\vspace{-0.3cm}
    \caption{AC vs $d_\text{TL}$ for different number of RIS reflecting elements.}
    \label{fig:AC_d_TL}
\end{figure}
The impact of the distance between the transmitting and listening tags, TT and LT, on the average channel capacity within RIS-aided BTTNs is explored in Fig. \ref{fig:AC_d_TL}. The imapct of different configurations (with and without RIS) is demonstrated, emphasizing the RIS's role in enhancing communication capabilities across varying distances. Notably, with an increased number of RIS elements, there is a substantial improvement in $\bar{C}$, especially noticeable at longer distances, highlighting the efficacy of RIS in maintaining robust communication in BTTNs. 
For instance, utilizing $40$ RIS elements enables maintaining an average capacity of $10 \ \mathrm{bit/Hz}$, a commendable achievement for extended distances typical in BTTN environments where the usual interaction range between TT and LT seldom surpasses $1 \text{m}$. Additionally, the rate of capacity reduction as the distance $d_{\mathrm{TL}}$ increases becomes significantly milder with an enhanced count of RIS elements, indicating their efficacy in mitigating the detrimental effects of longer distances. Illustratively, elevating the distance from $1\ \text{m}$ to $4\ \text{m}$ leads to an over $80\%$ drop in the average data rate in \textit{Without RIS} setups. In contrast, this reduction is limited to about $23\%$ when employing an RIS configuration with $40$ elements for the same distance.
These findings underscore the significant potential of RIS technology to extend the operational range of BTTNs while securing higher data rates, thereby supporting the deployment of RIS in low-power environments where direct transmission paths suffer from severe fading and other deleterious effects.

\section{Conclusion } \label{sec_conclusion}
This paper introduced the integration of an indoor RIS into BTTN systems and studied their overall performance. Initially, the research deduced closed-form expressions for the received SNR at LT by employing a moment-matching technique. This approach facilitated the development of simple yet profoundly effective analytical frameworks, capturing the intricate dynamics inherent to this complex system configuration.
Subsequently, to evaluate the impact of RIS on AC, OP, and BTTN's coverage area, compact analytical formulations for both OP and AC were derived, reflecting the variations in system parameters. The findings indicated that the implementation of RIS substantially augments the performance of BTTNs compared to conventional setups lacking RIS enhancements.
Ultimately, Monte Carlo simulations were conducted to corroborate the analytical outcomes, confirming the precision of the proposed analytical models for the RIS-aided BTTN systems.

\section*{Acknowledgment}
{This work is supported in part by the Academy of Finland under Grants 345072 and 350464 and the European Union's Horizon 2022 Research and Innovation Programme under Marie Skłodowska-Curie Grant
No. 101107993.}


\begin{thebibliography}{1}


\bibitem {zero6G1}
S. Naser, L. Bariah, S. Muhaidat, and E. Basar, ``Zero-energy devices empowered 6G networks: Opportunities, key technologies, and challenges,'' \textit{IEEE Internet of Things Magazine}, vol. 6, no. 3, pp. 44--50, 2023.

\bibitem {BCFAS1}
F. R. Ghadi, M. Kaveh, K. K. Wong, and Y. Zhang, ``Performance Analysis of FAS-Aided Backscatter Communications,'' \textit{IEEE Wireless Communications Letters}, 2024.

\bibitem {BCRFID1}
A. C. Y. Goay, D. Mishra, and A. Seneviratne, ``ASK modulator design for passive RFID tags in backscatter communication systems,'' in \textit{Proc. 2022 IEEE 22nd Annual Wireless and Microwave Technology Conference (WAMICON)}, pp. 1--4, Apr. 2022.

\bibitem {BCArbit1}
F. R. Ghadi, F. J. Martin-Vega, and F. J. López-Martínez, ``Capacity of backscatter communication under arbitrary fading dependence,'' \textit{IEEE Transactions on Vehicular Technology}, vol. 71, no. 5, pp. 5593--5598, 2022.

\bibitem {BCSI1}
M. Kaveh, F. Rostami Ghadi, R. Jäntti, and Z. Yan, ``Secrecy performance analysis of backscatter communications with side information,'' \textit{Sensors}, vol. 23, no. 20, p. 8358, 2023.

\bibitem {BTTN-design-1}
Y. Karimi, A. Athalye, S. R. Das, P. M. Djurić, and M. Stanaćević, ``Design of a backscatter-based tag-to-tag system,'' in \textit{Proc. 2017 IEEE International Conference on RFID (RFID)}, pp. 6--12, May 2017.

\bibitem {design-BTTN-2}
J. Ryoo, J. Jian, A. Athalye, S. R. Das, and M. Stanaćević, ``Design and evaluation of `bttn': a backscattering tag-to-tag network,'' \textit{IEEE Internet of Things Journal}, vol. 5, no. 4, pp. 2844--2855, 2018.

\bibitem {design-BTTN-3}
A. Athalye et al., ``Analog front end design for tags in backscatter-based tag-to-tag communication networks,'' in \textit{Proc. 2016 IEEE International Symposium on Circuits and Systems (ISCAS)}, 2016.

\bibitem {BTTN-performance-1}
T. Lassouaoui, F. D. Hutu, Y. Duroc, and G. Villemaud, ``Performance evaluation of passive tag to tag communications,'' \textit{IEEE Access}, vol. 10, pp. 18832--18842, 2022.

\bibitem {passive-BTTN-1}
P. V. Nikitin, S. Ramamurthy, R. Martinez, and K. S. Rao, ``Passive tag-to-tag communication,'' in \textit{Proc. 2012 IEEE International Conference on RFID (RFID)}, pp. 177--184, Apr. 2012.

\bibitem {BTTN-CSI-1}
A. Ahmad, A. Athalye, M. Stanacević, and S. R. Das, ``Collaborative channel estimation in backscattering tag-to-tag networks,'' in \textit{Proc. 1st ACM International Workshop on Device-Free Human Sensing}, pp. 35--38, Nov. 2019.



\bibitem {Riku_DirPathInterf}
R. Biswas, M. U. Sheikh, H. Yiğitler, J. Lempiäinen, and R. Jäntti, “Direct path interference suppression requirements for bistatic backscatter communication system,'' {\em IEEE 93rd Vehicular Technology Conference}, pp. 1--5, 2022.

\bibitem {Riku_IEEERFID1}
J. Liao, X. Wang, K. Ruttik,  R. Jäntti, and D. T. Phan-Huy, “In-band Ambient FSK Backscatter Communications Leveraging LTE Cell-Specific Reference Signals,'' {\em IEEE Journal of Radio Frequency Identification}, 2023.



\bibitem {BjornsonCSI}
E. Björnson, H. Wymeersch, B. Matthiesen, P. Popovski, L. Sanguinetti, and E. de Carvalho, 
``Reconfigurable intelligent surfaces: A signal processing perspective with wireless applications,” \textit{IEEE Signal Process. Mag.}, vol. 39, no. 2, pp. 135–158, 2022.

\bibitem {RIS-PLS-SG}
M. Kaveh, Z.  Yan, and R. Jantti,``Secrecy Performance Analysis of RIS-Aided Smart Grid Communications,” \textit{IEEE Transactions on Industrial Informatics}, vol. 20, no. 4, pp. 5415-5427, 2024.


\bibitem{ref57}
Y. Liu, Y. Ye, and R. Q. Hu  ``Secrecy outage probability in backscatter communication systems with tag selection,'' \textit{IEEE Wireless Communications Letters}, vol. 10, no. 10, pp. 2190--2194, 2021.


\bibitem{V2V-RIS-TIV}
F. R. Ghadi, M. Kaveh, and D. Martín, ``Performance analysis of RIS/STAR-IOS-aided V2V NOMA/OMA communications over composite fading channels,'' \textit{IEEE Transactions on Intelligent Vehicles}, vol. 9, no. 1, pp. 279--286, 2024.

\bibitem{ref41}
D. Loku, et al, ``RIS-empowered ambient backscatter communication systems,'' \textit{IEEE Wireless Communications Letters}, vol. 12, no. 1, 2022.




\bibitem{ref3941}
D. Galappaththige, F. Rezaei, C. Tellambura, and S. Herath, ``Optimizing Passive Tag Performance with Reconfigurable Intelligent Surfaces in Bistatic Backscatter Networks,'' \textit{IEEE Transactions on Vehicular Technology}, 2024. 


 

\bibitem{Gamma_dist_varmean}
Y. Ai, A. P. Felipe, L. Kong, M. Cheffena, S. Chatzinotas, and B. Ottersten, ``Secure vehicular communications through reconfigurable intelligent surfaces,'' \textit{IEEE Transactions on Vehicular Technology}, vol. 70, no. 7, pp. 7272--7276, 2021.

\bibitem{Gamma_dist_sum}
S. Atapattu, R. Fan, P. Dharmawansa, G. Wang, J. Evans, and T. A. Tsiftsis, ``Reconfigurable intelligent surface assisted two–way communications: Performance analysis and optimization,'' \textit{IEEE Transactions on Communications}, vol. 68, no. 10, pp. 6552--6567, 2020.

\bibitem{Mom-Match}
S. Al-Ahmadi and H. Yanikomeroglu, ``On the approximation of the generalized-K distribution by a gamma distribution for modeling composite fading channels,'' \textit{IEEE Transactions on Wireless Communications}, vol. 9, no. 2, pp. 706--713, Feb. 2010.

\bibitem{FAS-BC-ISAC}
F. R. Ghadi, K. K. Wong, F. J. Lopez-Martinez, H. Shin, and L. Hanzo, ``Performance Analysis of FAS-Aided NOMA-ISAC: A Backscattering Scenario,'' \textit{arXiv preprint arXiv:2408.04724}, 2024.

\bibitem{Ghadi123}
F. R. Ghadi, K. K. Wong, W. K. New, H. Xu, R. Murch, and Y. Zhang, ``On performance of RIS-aided fluid antenna systems,'' \textit{IEEE Wireless Communications Letters}, 2024.

\bibitem{ref59}    
A P. Prudnikov,  et al, ``Integrals and series: more special functions,'' \textit{CRC press}, Vol. 3, 1986.





%\bibitem{BD_Eng_detector1}
%R., Reed, F.L. Pour, and D.S., Ha, ``An energy efficient RF backscatter modulator for IoT applications,'' \textit{In 2021 IEEE International Symposium on Circuits and Systems (ISCAS)}, pp. 1--5, 2021.









%\bibitem{Indoor_RIS_01}
%S. Kayraklik, et al. ``Indoor Measurements for RIS-Aided Communication: Practical Phase Shift Optimization, Coverage Enhancement, and Physical Layer Security,'' \textit{IEEE Open Journal of the Communications Society}, 2024. 

\bibitem{Indoor_RIS_02}
J. Yuan, O. Franek, H. Fang, and P. Popovski, ``Indoor RIS-Assisted Wireless System with Location-Based Reflective Patterns,'' \textit{IEEE Transactions on Communications}, 2024.


\bibitem{BCAuth}
P. Wang, Z. Yan, and K. Zeng, ``BCAuth: Physical layer enhanced authentication and attack tracing for backscatter communications,'' \textit{IEEE Transactions on Information Forensics and Security}, vol. 17, pp. 2818--2834, 2022.








\end{thebibliography}
\end{document}